\documentclass[letterpaper, 10 pt, conference]{ieeeconf}  

\IEEEoverridecommandlockouts                              

\usepackage{cite}

\usepackage{amsmath,amssymb,amsfonts,amsthm}
\usepackage{algorithm}
\usepackage{algorithmic}
\usepackage{graphicx}
\usepackage{textcomp}
\usepackage{booktabs}
\usepackage{xcolor}
\usepackage{array} 
\usepackage{multirow} 

\usepackage{pstricks}
\usepackage{float}
\usepackage{tablefootnote}

\newtheorem{theorem}{Theorem}

\title{\LARGE \bf
Transformer Temperature Management and Voltage Control in Electric Distribution Systems with High Solar PV Penetration
}

\author{Amirhossein Ghorbansarvi, Dakota Hamilton, Mads R. Almassalkhi, and Hamid R. Ossareh$^{1}$
\thanks{This material is based upon work supported by NIST (award number 70NANB22H162).}
\thanks{$^{1}$ The authors are with the Department of Electrical and Biomedical Engineering, University of Vermont, Burlington, VT, USA
{\tt\small \{aghorban, dhamilt6, malmassa, hossareh\}@uvm.edu}}%
}

\begin{document}

\maketitle
\thispagestyle{empty}
\pagestyle{empty}

\begin{abstract}

The increasing penetration of photovoltaic (PV) systems in distribution grids can lead to overvoltage and transformer overloading issues. While voltage regulation has been extensively studied and some research has addressed transformer temperature control, there is limited work on simultaneously managing both challenges. 
This paper addresses this gap by proposing an optimization-based strategy that efficiently manages voltage regulation and transformer temperature while minimizing the curtailment of PV generation. In order to make this problem convex, a relaxation is applied to the transformer temperature dynamics constraint.
We also provide analysis to determine under which conditions this relaxation remains tight. 
The proposed approach is validated through simulations, demonstrating its effectiveness in achieving the desired control objectives.

\end{abstract}

\section{Introduction}

The integration of renewable energy sources, such as solar photovoltaics (PV), into electric distribution systems represents a significant step towards sustainability and energy independence.
However, high PV penetration may also introduce significant challenges in the safe operation of the grid.
More specifically, excess power generated by distributed solar resources can cause reverse power flows which can lead to issues in voltage regulation and protection of distribution system components~\cite{Intro}. 
In particular, managing the loading of critical grid components, such as substation transformers, is increasingly important in this context as excessive overloading can accelerate degradation and potentially lead to catastrophic failures (with major financial implications). 
Advanced PV inverters can play a crucial role in mitigating these issues by curtailing active power generation and providing reactive power support. Nevertheless, minimizing PV curtailment is essential, as it represents a waste of clean, renewable energy that could reduce reliance on fossil fuels and lower greenhouse gas emissions~\cite{aleem2020}.
Thus, in this paper, we explore methods for managing distribution system voltages and transformer temperatures while minimizing curtailment of renewable PV generation.

Existing voltage control methods which leverage distributed generation (DG) resources can mainly be classified into two categories: centralized and decentralized strategies.
Centralized control strategies, such as those presented in~\cite{Cent}, focus on managing voltage constraints in active distribution systems through coordinated operation of DG resources. 
These approaches often aim to minimize system losses and reduce voltage deviations at each bus~\cite{senjyu2008}.
Centralized control is effective for relatively long-term operations (e.g. hourly), but it may not be practical for real-time operations due to its complex communication and computation requirements.

Under decentralized voltage control, DGs are managed locally, by their own controllers instead of a central one. Thus, decentralized schemes offer lower computational complexity compared to centralized voltage control. 
Various works in the literature have proposed decentralized control strategies for preventing voltage violations in distribution networks including droop-based active power curtailment~\cite{droop} and reactive power management~\cite{kundu,jahangiri2013}.
However, decentralized approaches often suffer from suboptimality due to a lack of system-wide information and coordination among controllers.
For similar reasons, they also typically can not guarantee that grid operational constraints are satisfied. 

Voltage regulation techniques have also been coupled with advanced control methods, such as model predictive control (MPC). 
In contrast to conventional control methods, MPC handles constraints in system states, inputs, and outputs by solving an optimization problem to determine control inputs.
MPC uses a mathematical model of system dynamics within this optimization problem which can incorporate forecasts of solar generation and demand. 
For example, an MPC is developed using a PV-based reactive power management scheme to minimize power loss and stabilize voltage fluctuations in~\cite{mpc1}. 
Similarly,~\cite{vancutsem2013} presents a centralized MPC-based controller designed to optimally regulate the output of DGs, including both active and reactive power, with the goal of maintaining monitored voltages within specified target ranges determined by security or economic criteria.
MPC-based schemes have also been proposed which leverage energy storage systems in addition to PV inverters to further reduce voltage fluctuations and curtailment~\cite{bat}.
Finally, in situations where an accurate system model is unavailable, data-driven methods can offer viable solutions that do not require a complete distribution network model~\cite{datadriven}. 

As mentioned earlier, power transformers are crucial components of any power system, which can suffer considerable financial and operational losses due to grid failures. Extended periods of overloading can accelerate degradation and increase the risk of serious damage to distribution system equipment, including transformers. For this reason, there is interest in approaches that actively regulate transformer loading, including dynamic transformer (temperature) rating, where the transformer hot-spot temperature is managed~\cite{EV1Temp,EV2Temp,temp1Ref,tempMads,hussein1}.


While voltage regulation and temperature management have been studied individually, few works have proposed a unified framework that addresses both challenges simultaneously. Moreover, many existing centralized MPC approaches achieve good results, but at the expense of large prediction horizons. This increases the computational burden, making real-time implementation challenging, especially in large-scale systems. 

Building on these gaps, our work proposes a centralized MPC strategy that integrates both voltage regulation and transformer temperature control while minimizing curtailment in high renewable penetration scenarios. Specifically, the choice of the objective function in our work allows for effective performance even with a very small prediction horizon. This significantly reduces the computational burden while still achieving good results in terms of minimizing curtailment and maintaining system reliability. We show that our optimization problem can be convexified by relaxing the transformer temperature model, which is quadratic in the transformer apparent power~\cite{temp}. Finally, we provide mathematical guarantees, using Karush-Kuhn-Tucker (KKT) analysis, to ensure the tightness of this relaxation under certain conditions. 

In summary, the main contributions of this work are:
\begin{enumerate}
    \item A centralized MPC strategy is proposed to simultaneously manage voltage magnitudes and substation transformer temperature while minimizing curtailment in distribution networks with high renewable penetration. 
    \item Non-convex constraints associated with the transformer temperature dynamics are relaxed to beget a convex formulation. KKT analysis then explores conditions under which this relaxation is tight.
    \item The closed-loop performance of the MPC is examined via numerical simulations for varying prediction horizon lengths. Trade-offs between MPC performance and computation time are also evaluated.  
\end{enumerate}

The remainder of the paper is organized as follows: 
Section~II discusses modeling of the distribution network, inverters, and substation transformer temperature. 
These models inform constraints of the optimization problem within the MPC, which is described in Sec.~III. 
In Sec.~IV, a relaxation of the optimization problem is introduced, and conditions under which it is tight are explored using KKT analysis. 
Numerical case studies and simulation results are presented in Sec.~V. 
Section~VI concludes and discusses future work.

\section{Distribution System Modeling}
\label{sec:modeling}

In this section, we describe the models of the distribution network and relevant components that will be used in the proposed MPC approach.

\subsection{Network model}

Consider a connected, radial distribution network with $N+1$ nodes and no lateral branches, as depicted in Fig.~\ref{fig:dist_network}. 
We denote the set of nodes in the network by $\mathcal{N}=\{0, 1, \ldots, N\}$, where node~$0$ represents the substation. 
The set of discrete time instances of interest are denoted by $t \in $ \(\mathcal{T} = \{1, \ldots, T\}\). 
We use the LinDistFlow equations to model power flow in the distribution network~\cite{kekatos,taheri}. {LinDistFlow}, which is linearized and derived from the DistFlow equations~\cite{baran} by assuming that line losses are negligible, is expressed in vector form as:
\begin{gather}
\mathbf{v}(t) = \mathbf{R}\mathbf{p}(t) + \mathbf{X}\mathbf{q}(t) + v_0(t)\mathbf{1}_N\,, \label{linVol} \\
\mathbf{p}(t) = \mathbf{A}^\top \mathbf{P}(t)\,, \\
\mathbf{q}(t) = \mathbf{A}^\top \mathbf{Q}(t)\,,
\end{gather}
where the vectors \(\mathbf{v}(t)=[v_1(t), \ldots, v_N(t)]^\top\), \(\mathbf{p}(t) = [p_1(t), \ldots, p_N(t)]^\top\), and \(\mathbf{q}(t) = [q_1(t), \ldots, q_N(t)]^\top\) consist of bus voltage magnitudes, active power injections, and reactive power injections at each node, respectively. 
The vectors of real and reactive power flows in each branch are denoted by $\mathbf{P}(t)$ and $\mathbf{Q}(t)$, respectively.
We denote the $N \times 1$ vector of ones by \(\mathbf{1}_N\), and $v_0$ is the substation voltage magnitude.
The matrices \(\mathbf{R} = \mathbf{F}\operatorname{diag}(\mathbf{r})\mathbf{F}^\top\) and \(\mathbf{X} = \mathbf{F}\operatorname{diag}(\mathbf{x})\mathbf{F}^\top\), where \(\mathbf{r} \in \mathbb{R}^N\) and \(\mathbf{x} \in \mathbb{R}^N\) are vectors of line resistances and reactances, respectively. Here, \(\operatorname{diag}(\mathbf{y})\) denotes a diagonal matrix with the elements of the vector $\mathbf{y}$ on its main diagonal. 
The matrix \(\mathbf{F} =\mathbf{A}^{-1}\), where \(\mathbf{A}\) is the reduced branch-bus incidence matrix~\cite{incidence1}. 


\subsection{Inverter model}

We assume that each bus of the network may contain both load and an inverter-interfaced PV that can regulate its active and reactive power outputs.
The active and reactive power consumed by the load at each node $j$ and time $t$ is denoted by $p_j^c(t)$ and $q_j^c(t)$, and the active and reactive power output of the inverter is given by $p_j^g(t)$ and $q_j^g(t)$. 
The active power curtailment of each inverter is denoted by $p_j^{cr}$. 
Thus, the net nodal injection at each bus can be expressed as
\begin{align}
    p_j(t) &= p_j^g(t) - p_j^{cr}(t) - p_j^c(t)\,, & \forall j \in \mathcal{N}\,, \\
    q_j(t) &= q_j^g(t) - q_j^c(t)\,, & \forall j \in \mathcal{N}\,.
\end{align}
 Furthermore, the reactive power capability of each inverter is limited by its fixed apparent power capability \(s_{j,\text{max}}\), i.e.,
\begin{equation}
    \left[q_j^g(t)\right]^2 \leq s_{j,\text{max}}^2 - \left[p_j^g(t) - p_j^{cr}(t)\right]^2\,, \quad \forall j \in \mathcal{N}\,. \label{eq:inv_app_power}
\end{equation}

\begin{figure}
\centering
\includegraphics[trim={0 2.5cm 0 0},clip,width=0.9\linewidth]{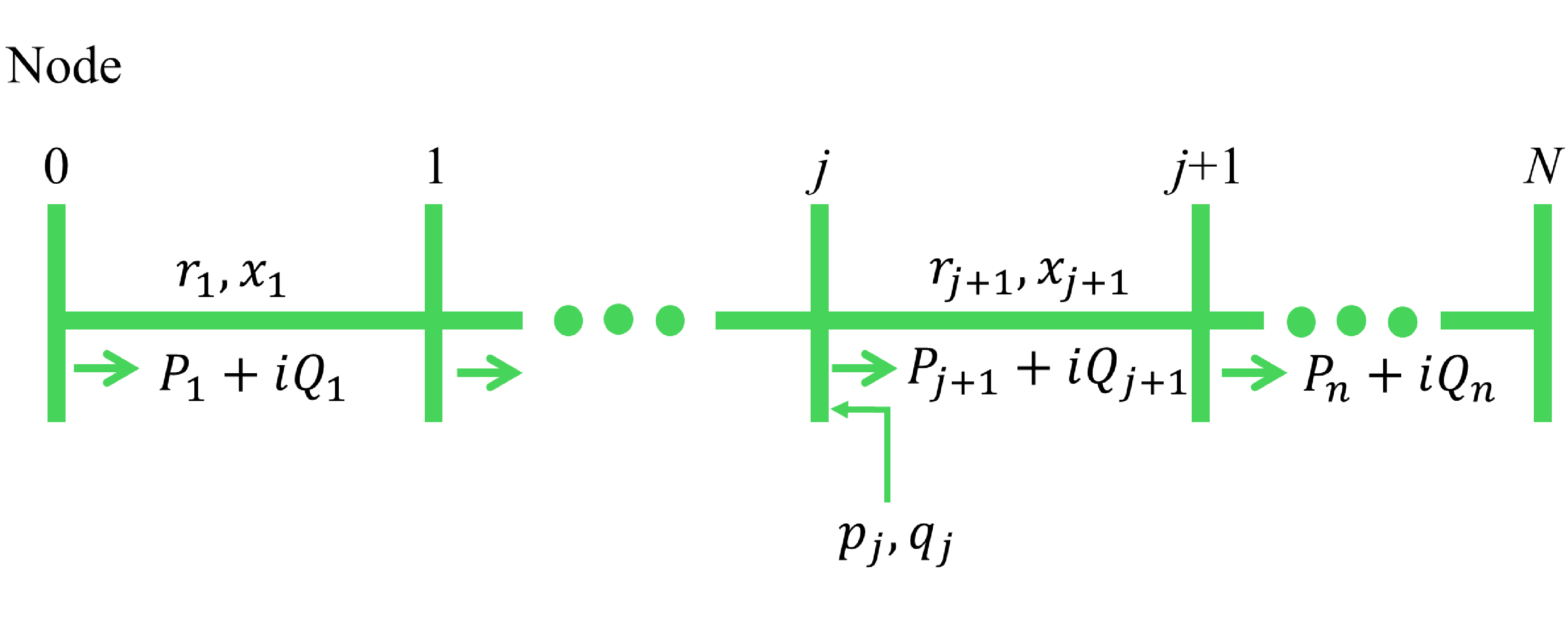}
\caption{Illustration of radial distribution network.}
\label{fig:dist_network}
\end{figure}

\subsection{Transformer temperature model}

In this paper, we model the substation transformer temperature using the following regression model~\cite{temp}, based on experimental data from field tests:
\begin{equation}
T(t+1)=aT(t) + b\left[ P^2_{\text{total}}(t) + Q^2_{\text{total}}(t) \right] + cT_a(t) + d\,, \label{eq:temp_model}
\end{equation}
where the hot-spot temperature at the next timestep, $T(t+1)$, depends on the current temperature, $T(t)$, the total active and reactive power flowing through the substation transformer, $P_{\text{total}}(t)$ and $Q_{\text{total}}(t)$, and the ambient temperature, $T_a(t)$, at time $t$. 
Expressions for $P_{\text{total}}(t)$ and $Q_{\text{total}}(t)$ are given by:
\begin{gather}
    P_{\text{total}}(t) = \sum_{j=1}^N p_j(t)\,, \\
    Q_{\text{total}}(t) = \sum_{j=1}^N q_j(t)\,.
\end{gather}
The constant coefficients
$a$, $b$, $c$ and $d$ are calculated via regression on experimental data.
Note that all coefficients are positive and have values less than one. 
\begin{figure}
\centering
    \includegraphics[width=0.95\linewidth]{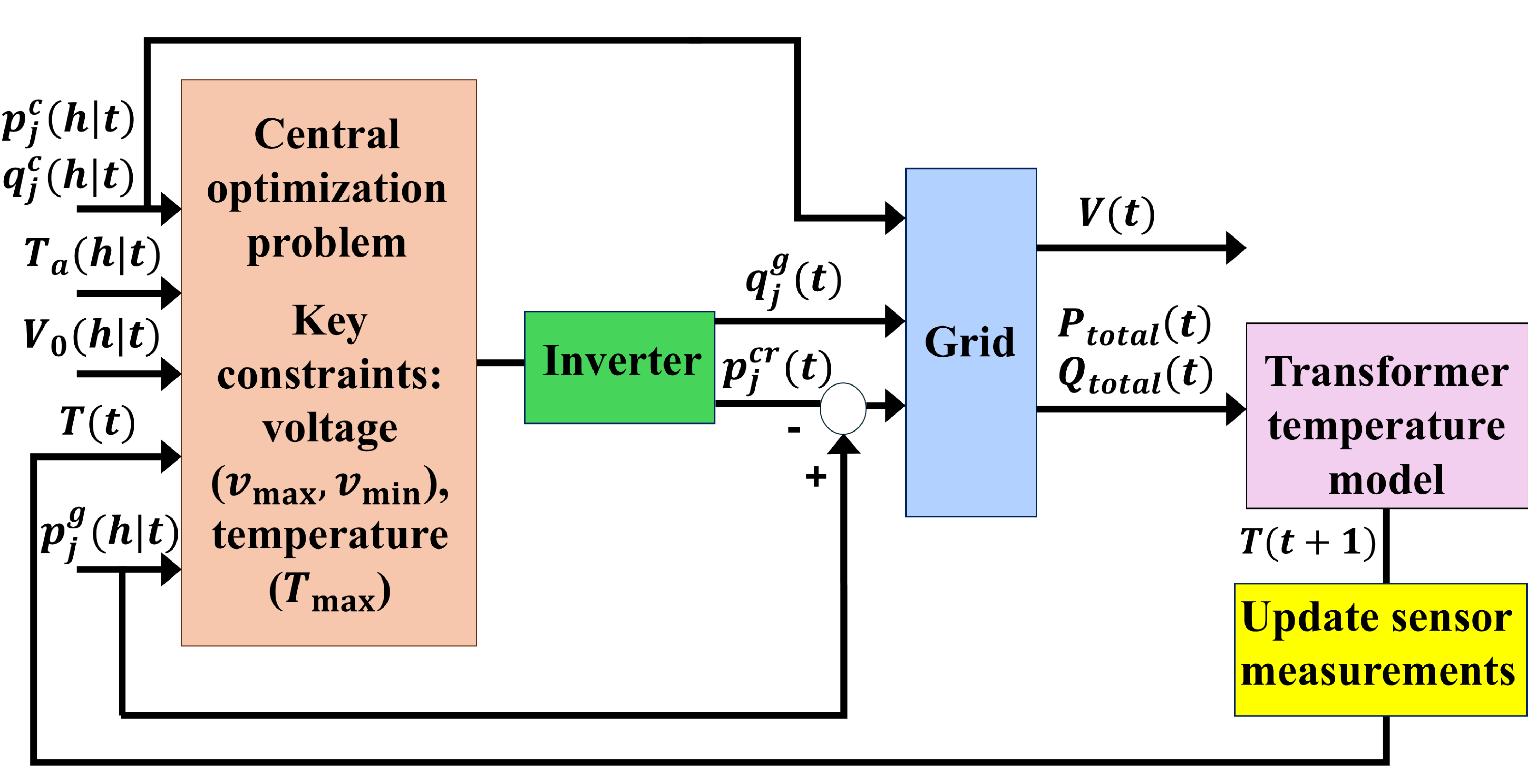}
\caption{Block diagram of centralized MPC framework.}
\label{fig:mpc_block_diagram}
\end{figure}
\section{MPC Framework and Problem Formulation}
Our goal in this work is to design a controller for distribution systems which minimizes the curtailment of solar PV while ensuring nodal voltages and the substation transformer temperature stay within prescribed limits. 
To this end, we introduce a centralized, MPC-based framework as illustrated in Fig.~\ref{fig:mpc_block_diagram}.

It is envisioned that, at each time step $t\in \mathcal{T}$, the proposed MPC would receive updated measurements from temperature sensors in the substation transformer as well as forecasts of relevant input data over a finite prediction horizon of $H$ steps.
More specifically, this input data would include forecasts of: i) active and reactive power consumption at each node, $ p_j^c(h|t)$ and $q_j^c(h|t)$; ii) available active power generation at each node, $p_j^{g}(h|t)$; iii) substation voltage magnitude, $v_0(h|t)$; and iv) ambient temperature, $T_a(h|t)$.\footnote{In this work, we assume that the centralized MPC receives perfect forecasts. Studying the impact of forecast uncertainty is our future work.}
Here, we use the notation $(h|t)$ to denote the predicted value $h$ steps ahead of the present time $t$, where $h \in\mathcal{H} = \{0, 1, \ldots, H-1\}$.


Given these inputs, the MPC determines the optimal active power curtailment and reactive power output decisions for each inverter by solving the following optimization problem:

\begin{subequations}
\begin{align}
\text{(P1)}~&\min_{ p_j^{cr}(h),q_j^g(h)} \sum_{h=0}^{H-1} \sum_{j=1}^{N} \Bigl[ \beta(p_j^{cr}( h))^2 +  (q_j^{g}( h))^2
  \Bigr]\,,\label{eq:obj_fun}\\
    \text{s.t.} ~~
    &p_j(h) = p_j^g(h|t)  - p_j^{cr}(h)  - p_j^c(h|t)\,,\nonumber \\&\hspace{4.25cm}\forall j\in\mathcal{N}\,, h \in\mathcal{H}\,,\label{eq:p_con}\\
    &q_j(h)= q_j^g(h) - q_j^c(h|t)\,,~~~~\forall j\in\mathcal{N}\,, h \in\mathcal{H}\,, \\
    &v_j(h) = \sum_{i=1}^{N}R_{ji}p_i(h) 
    + \sum_{i=1}^{N}X_{ji}q_i(h) + v_0(h|t)\,, \nonumber \\ &\hspace{4.25cm}\forall j\in\mathcal{N}\,, h \in\mathcal{H}\,, \label{eq:vol_equality} \\
    &T(h+1) = aT(h)+ b\left[ P^2_{\text{total}}(h) +Q^2_{\text{total}}(h) \right] \nonumber \\&\hspace{2.0cm} + cT_a(h|t) + d\,,~~~~~~~\forall h \in\mathcal{H}\,, \label{eq:temp_con} \\
    &P_{\text{total}}(h) = \sum_{j=1}^{N} p_j(h) \,,~~~~~~~~~~~~~~~~~\forall h \in\mathcal{H}\,,\\
     &Q_{\text{total}}(h)= \sum_{j=1}^{N} q_j(h) \,,~~~~~~~~~~~~~~~~\forall h \in\mathcal{H}\,,\label{eq:Qtot_con}\\
     &T(h+1)\le T_{\text{max}}\,,~~~~~~~~~~~~~~~~~~~~~\forall h \in\mathcal{H}\,, \label{eq:uptemp} \\
     &v_{\text{min}} \le v_j(h) \le v_{\text{max}} \,,~~~~~~~~\forall j \in \mathcal{N}\,, h \in\mathcal{H}\,, \label{eq:vol_inequality} \\
     &0 \leq p_j^{cr}(h) \leq p_j^g(h|t)\,,~~~~~~~\forall j \in \mathcal{N}\,, h \in\mathcal{H}\,, \label{eq:curt_con}\\
    &\left[q_j^g(h)\right]^2 \leq s_{j,\text{max}}^2 - \left[p_j^g(h|t) - p_j^{cr}(h)\right]^2\,, \nonumber \\& 
 \hspace{4.25cm}\forall j \in \mathcal{N}\,, h \in\mathcal{H}\,. \label{eq:inverter_S}
\end{align}
\end{subequations}
The first term of the objective function~\eqref{eq:obj_fun} serves to minimize PV curtailment across all nodes and all steps of the prediction horizon. 
This is the primary objective of the controller; however, under certain situations (e.g., when voltage constraints are not binding), the optimization problem (P1) can have non-unique solutions because the reactive power generation at each node can become a free variable.
This may lead to larger-than-necessary reactive power injections at each node, which increases the transformer temperature and degrades closed-loop performance of the MPC when the length of the prediction horizon, $H$, is too short.
Thus, in order to improve closed-loop performance, we introduce a second term to the objective that penalizes $q_j^g(h)$ as well as a positive scaling constant, $\beta$.
Further justification for including this second term and a discussion of tuning the parameter $\beta$ are provided in Sec.~\ref{sec:no_curt}.

The equality constraints~\eqref{eq:p_con}--\eqref{eq:Qtot_con} capture the LinDistFlow power flow equations and transformer temperature model as described in Sec.~\ref{sec:modeling}, where $R_{ji}$ and $X_{ji}$ denote the elements in the $j$-th row and $i$-th column of the matrices $\mathbf{R}$ and $\mathbf{X}$, respectively.
Note that the initial transformer temperature $T(0)$ in the MPC model (i.e., at step $h=0$) is updated at each time $t$ to reflect the most recent temperature sensor measurement.
Additionally, we emphasize that is important to ensure that the initial temperature at $t=0$ (i.e., at the beginning of the simulation), denoted \( T_0 \), is less than the upper bound, \( T_{\text{max}} \); otherwise, the optimization problem may be infeasible.
An upper bound on the transformer temperature, $T_\text{max}$, is enforced at each step of the prediction horizon by the inequality constraint~\eqref{eq:uptemp}, and~\eqref{eq:vol_inequality} sets upper and lowers bounds on bus voltage magnitudes, $v_\text{min}$ and $v_\text{max}$.
Finally, the inequality constraints~\eqref{eq:curt_con}--\eqref{eq:inverter_S} enforce inverter apparent power limits and ensure that PV curtailment does not exceed available PV generation at each node.

The solution of the optimization problem (P1) provides a trajectory of optimal curtailment and reactive power output setpoints for each inverter (i.e., $p_j^{cr}(h)$ and $q_j^g(h)$ for all $h\in\mathcal{H}$). 
The first of these setpoint decisions (i.e., $p_j^{cr}(0)$ and $q_j^g(0)$) is then sent to each inverter $j$ at time $t$, and the optimization problem is re-solved at time $t+1$ with updated inputs.

Note that the optimization problem formulation (P1) is non-convex due to the quadratic equality constraint~\eqref{eq:temp_con}.
This presents a challenge for real-time operation of the MPC as methods for solving non-convex programs can be computationally expensive.
Moreover, this non-convexity means that we can not guarantee that solutions of (P1) are globally optimal.
Thus, in the next section, we introduce a convex relaxation of the problem and explore conditions under which this relaxation is tight.

\section{Convex Relaxation and Tightness Guarantee}
In order to make the problem convex, we define an auxiliary variable $e(h)$, and replace the temperature dynamics constraint~\eqref{eq:temp_con} with the following relaxation:
\begin{gather}
T(h+1)=aT(h) + b\left[ e(h) \right] + cT_a(h|t) + d\,, \label{eq:relaxation1}\\
    e(h) \ge P^2_{\text{total}}(h) + Q^2_{\text{total}}(h)\,, \label{eq:relaxation2}
\end{gather}
for all $h\in\mathcal{H}$.
Note that this relaxation provides a more conservative prediction of the temperature dynamics.
That is, the transformer temperature trajectories predicted by~\eqref{eq:relaxation1}--\eqref{eq:relaxation2} will always be greater than or equal to the temperatures predicted by~\eqref{eq:temp_con} for the same values of $P_\text{total}(h)$, $Q_\text{total}(h)$, and $T_a(h|t)$.

When~\eqref{eq:relaxation2} holds with equality (i.e., the relaxation is tight), then the solution of the relaxed (convex) problem is a feasible and globally optimal solution of the original problem.
Thus, we are interested in understanding conditions under which the relaxation is tight. 
In the remainder of this section, we use KKT analysis to prove such conditions for the case when bus voltage magnitude constraints are not binding.
We leave analysis of the case when voltage magnitudes are binding as future work.

\begin{theorem}
\label{thm:main}
If, at optimality, there exists a time step $h^*$ and a node $j^*$ such that $p_{j^*}^{cr}(h^*) > 0$, and the voltage constraints \eqref{eq:vol_inequality} are not binding for all $j\in\mathcal{N}$ at time $h^*$, then the relaxation is tight for all \( h \in \{0, 1, 2, \ldots, h^*\} \).
\end{theorem}

\begin{proof}
For every \( h \in \mathcal{H} \) and \( j \in \mathcal{N} \), the following constraints define primal feasibility and dual variables:
\begin{subequations}
\begin{align}
    0 &= T(h+1) - aT(h) - be(h) \nonumber \\
      &\quad - cT_a(h|t) - d \,,
      && \nu_T^{h+1} \in \mathbb{R} \quad \forall h \,,
\end{align}
\begin{flalign}
    0 &= P_{\text{total}}(h) - \sum_{j=1}^{N} p_j(h) \,,
    && \nu_P^{h} \in \mathbb{R} \quad  \forall h \,,
\end{flalign}
\begin{flalign}
    0 &= Q_{\text{total}}(h) - \sum_{j=1}^{N}q_j(h)\,,
    && \nu_Q^{h} \in \mathbb{R} \quad \forall h \,,
\end{flalign}
\begin{flalign}
    T(h+1) - T_{\max} &\le 0\,,
    && \lambda_T^{h+1} \ge 0 \quad  \forall h \,,
\end{flalign}
\begin{flalign}
    p_j^{cr}(h) - p_j^g(h|t) &\le 0\,,
    && \overline{\lambda}_j^{h} \ge 0 \quad \forall h, \forall j \,,
\end{flalign}
\begin{flalign}
    -p_j^{cr}(h) &\le 0\,,
    && \underline{\lambda}_j^{h} \ge 0  \quad \forall h, \forall j \,,
\end{flalign}
\begin{align}
   \hspace{-1cm} P^2_{\text{total}}(h) &+ Q^2_{\text{total}}(h) - e(h) \leq 0 \,,
    && \lambda_e^{h} \geq 0 \quad \forall h \,, \label{eq:e_dual}
\end{align}
\begin{flalign}
    \left[q_j^g(h)\right]^2 - s_{j,\text{max}}^2 \nonumber \\
    + \left[p_j^g(h|t) - p_j^{cr}(h)\right]^2 \leq 0 \,,
    && \lambda_{s,j}^{h} \geq 0 \quad \forall h, \forall j \,. \label{eq:inv_app_power_dual}
\end{flalign}
\end{subequations}
Here, \(\lambda\) and \(\nu\) represent the dual variables for inequality and equality constraints, respectively. 
Moreover, since the voltage constraints are assumed to be non-binding, we omit these equations as their corresponding dual variables are zero (due to complementary slackness). The general form of stationarity condition $\nabla_{y(h)} \mathcal{L}(y, \lambda, \nu) = 0 $ has to hold for each variable \(y\) at timestep \(h\), which gives:
\begin{subequations}
\begin{align}
   \nabla_{T(h+1)} \mathcal{L} =0 &\implies \nonumber \\
   &\hspace{-1.5cm}\nu_T^{h+1} = a\nu_T^{h+2} - \lambda_T^{h+1}, \quad \forall h \in \mathcal{H} \setminus \{H-1\} \,,\label{eq:nu_T} \\
  \nabla_{T(H)} \mathcal{L}=0 &\implies \nu_T^{H} =  - \lambda_T^{H} \label{eq:last_nu_T}\,, \\
    \nabla_{p_j^{cr}(h)} \mathcal{L}=0 &\implies
    2\beta p_j^{cr}(h) + \overline{\lambda}_j^{h} - \underline{\lambda}_j^{h} +  \nu_P^h \nonumber \\
   &\hspace{-1.0em}- 2\lambda_{s,j}^h\left(p_j^{g}(h|t)- p_j^{cr}(h)\right)= 0\,, 
    \forall h, \forall j \label{eq:grad_Pcr} \\
    \nabla_{e(h)} \mathcal{L}=0 &\implies -b\nu_T^{h+1} - \lambda_e^{h} = 0\,, \quad \quad \forall h\,, \label{grad_E} \\
    \nabla_{P_{\text{total}}(h)} \mathcal{L}=0 &\implies 2\lambda_e^h P_{\text{total}}(h) + \nu_P^h = 0\,,~  \forall h\,, \label{grad_Ptotal} \\
    \nabla_{Q_{\text{total}}(h)} \mathcal{L}=0 &\implies 2\lambda_e^h Q_{\text{total}}(h) + \nu_Q^h = 0\,,~ 
     \forall h\,, \label{grad_Qtotal} \\
    \nabla_{q_j^{g}(h)} \mathcal{L}=0 &\implies \nonumber \\&\hspace{-1.5em} 2q_j^{g}(h) - \nu_Q^h + 2\lambda_{s,j}^hq_j^{g}(h) = 0\,,~   \forall h,\forall j \,.\label{eq:grad_qg}
\end{align}
\end{subequations}
Note that~\eqref{eq:last_nu_T} is a special case of~\eqref{eq:nu_T} for the last timestep in the prediction horizon. 

From the complementary slackness conditions, if $\lambda_e^h > 0$, then~\eqref{eq:e_dual} must hold with equality (and thus the relaxation is tight) at time $h$.
Therefore, in order to show the relaxation is tight for all $h\le h^*$, we need to show that $\lambda_e^h > 0$ for all $h\le h^*$.

First, we will show that for all $h \le h^*$, the dual variable $\lambda_e^h \ge \lambda_e^{h^*}$. 
Therefore, if we prove that $\lambda_e^{h^*}$ is strictly positive, we can also conclude that for all $h \le h^*$, the relaxation is tight.
 From recursion on \eqref{eq:nu_T} and \eqref{eq:last_nu_T}, we have
 \begin{equation}
    \nu_T^{h+1} = - \sum_{k=h+1}^{H} a^{k-h-1}\lambda_T^k\,. \label{summation_nu}
\end{equation}
\\
Furthermore, substituting~\eqref{summation_nu} into~\eqref{grad_E}, we obtain
\begin{equation}
    \lambda_e^{h} = b  \sum_{k=h+1}^{H} a^{k-h-1} \lambda_T^k\,. \label{lamda_e_sum}
\end{equation}
Next, for all $h \le h^*$, we break up the summation in~\eqref{lamda_e_sum} as
 \begin{equation} 
  \lambda_e^{h} = b \sum_{k=h+1}^{h^*} a^{k-h-1} \lambda_T^k + b\sum_{k=h^*+1}^{H} a^{k-h^*-1} \lambda_T^k\,. \label{lamda_E_L}
\end{equation}
Note that the second term in~\eqref{lamda_E_L} is equal to $\lambda_e^{h^*}$ based on~\eqref{lamda_e_sum}.
Additionally, since $a$, $b$, and $\lambda_T^k$ are non-negative by definition, then the first term in~\eqref{lamda_E_L} is also non-negative. 
Thus, it is clear that for all $h \leq h^*$, $\lambda_e^{h} \ge \lambda_e^{h^*}$. 

Next, we prove that $\lambda_e^{h^*}$ is strictly positive when the assumptions in the theorem statement hold. 
More specifically, we consider the case where $\lambda_e^{h^*} = 0$ at optimality and will show that this leads to a contradiction. 
Substituting $\lambda_e^{h^*} = 0$ into~\eqref{grad_Ptotal} and~\eqref{grad_Qtotal}, we have
\begin{gather}
     2\lambda_e^{h^*} P_{\text{total}}(h^*) + \nu_P^{h^*} = 0  \implies \nu_P^{h^*} = 0\,, \\
     2\lambda_e^{h^*} Q_{\text{total}}(h^*) + \nu_Q^{h^*} = 0  \implies \nu_Q^{h^*} = 0\,.
\end{gather}
Then, substituting $\nu_Q^{h^*} = 0$ into~\eqref{eq:grad_qg}, we obtain
\begin{equation}
2q_j^{g}(h^*) + 2\lambda_{s,j}^{h^*}q_j^{g}(h^*) = 0\,,  \label{cont_Lambda_s}
\end{equation}
which implies that $q_j^{g}(h^*) = 0$ for all $j\in\mathcal{N}$ (since $\lambda_{s,j}^{h^*} \ge 0$). Furthermore, consider the case when 
 $\lambda_{s,j^*}^{h^*} > 0$, where $j^*$ is the node with positive PV curtailment (by the theorem statement).
 Due to complementary slackness, this implies that~\eqref{eq:inv_app_power_dual} holds with equality.
 This means that, with $q_{j^*}^{g}(h^*) = 0$, we would have
\begin{equation}
     s_{j^*,\text{max}} = p_{j^*}^g(h|t) - p_{j^*}^{cr}(h) \implies p_{j^*}^g(h|t) > s_{j^*,\text{max}}\,.
\end{equation}
However, this contradicts our modeling assumption that the inverter apparent power rating is always greater than or equal to the available active power generation.\footnote{Note that the available generation from a PV panel may exceed the apparent power rating of its inverter in practice (e.g., if the inverter is undersized). However, this excess power would always be curtailed and would not be available from the grid operator's perspective. Thus, it would not make sense to have $p_{j^*}^g(h|t) > s_{j,\text{max}}$ in this context.} 
Therefore, we conclude that, under these conditions,~\eqref{eq:inv_app_power_dual} must not be binding at node $j^*$ and its corresponding dual variable, $\lambda_{s,j^*}^{h^*}$, must be zero.

Finally, replacing $\nu_P^{h^*} = 0$, $\lambda_{s,j^*}^{h^*}=0$ into~\eqref{eq:grad_Pcr}, we have
\begin{equation}
    2\beta p_{j^*}^{cr}(h^*) + \overline{\lambda}_{j^*}^{h^*} - \underline{\lambda}_{j^*}^{h^*} = 0\,. \label{contradiction1}
\end{equation}
Recall that the theorem statement assumes $p_{j^*}^{cr}(h^*) > 0$. 
Moreover, due to complementary slackness, we have $\underline{\lambda}_j^{h^*}=0$.
Thus, the left hand side of~\eqref{contradiction1} is strictly positive, which is a contradiction. Therefore, we conclude that \(\lambda_e^{h^*} > 0\).
Thus, \(\lambda_e^{h} \ge \lambda_e^{h^*} > 0\) and the relaxation is tight for all \(h \le h^*\).
\end{proof}

Theorem~\ref{thm:main} provides specific conditions under which the relaxation in~\eqref{eq:relaxation2} is tight. 
In particular, we considered the case when voltage constraints are not binding and PV curtailment is required. 
This assumption may be reasonable in some practical situations since PV curtailment to manage excessive transformer temperature rise can often also help alleviate overvoltage issues.
However, the authors recognize that these assumptions may not always hold, and that voltage constraints may be binding at optimality.
Thus, in the next section, we explore the performance of the proposed control architecture through numerical case studies, including both scenarios with and without binding voltage constraints.

\section{Numerical Case Study Results}

In order to evaluate the effectiveness of the proposed MPC framework in managing transformer temperatures and voltages, we conducted numerical simulations on a radial distribution network.
The test feeder consists of 6 buses with no lateral branches.
The consumption at each node is modeled as a constant PQ load, and is randomly generated at each time $t\in\mathcal{T}$.
A solar PV panel and inverter are also installed at the end of the feeder (i.e., the last node).
The physical plant consists of the grid and transformer temperature model~\eqref{eq:temp_model}, and is modeled in MATLAB.
We emphasize that while the optimization model in the MPC uses the LinDistFlow approximation of the power flow equations, the plant model uses the full AC power flow equations, which are implemented using MATPOWER~\cite{matpower}.
The optimization problem (P1) is implemented in MATLAB using the CVX framework and solved using the Gurobi solver. 
Parameter values used in the numerical study are summarized in Table~\ref{tab:NumericalVals}.
\begin{table}
\centering
\vspace{1.0em}
\caption{ Parameter Values for Numerical Case Study}  
\begin{tabular}{c>{\centering\arraybackslash}p{5.3cm}>{\centering\arraybackslash}p{1cm}}
    \toprule
    \textbf{Variable} & \textbf{Description} & \textbf{Value} \\
    \midrule
    $S_{\text{base}}$  & Per unit base apparent power (MVA) & 2.5 \\
    $V_{\text{base}}$  & Per unit base voltage (kV) & 4.8 \\
     $v_{\text{max}}$  & Maximum voltage constraint (p.u.) & 1.05  \\ 
    $v_{\text{min}}$  & Minimum voltage constraint (p.u.) & 0.95  \\
    $T_{\text{max}}$  & Maximum temperature constraint ($^\circ$C) & 56  \\ 
    $T_0$ & Transformer initial temperature at $t=0$ ($^\circ$C) & 35  \\
    $T_a(h|t)$ & Ambient temperature\tablefootnote{We assume the ambient temperature is constant throughout the numerical simulations.} ($^\circ$C) & 35  \\
    $a$ & Temperature model coefficient (unitless) & 0.9972\\
    $b$ &  Temperature model coefficient ($^\circ$C\,/\,MVA$^2$) & 0.0241\\
    $c$ & Temperature model coefficient (unitless) & 0.0005\\
    $d$ & Temperature model coefficient ($^\circ$C) & 0.0931\\
    \bottomrule
    \label{tab:NumericalVals}
\end{tabular}
\end{table}

\subsection{Voltage and temperature issues without PV curtailment}

In the following case studies, we consider a high PV penetration scenario, where solar generation significantly exceeds local demand during peak sunlight hours.
The load at each node is randomly generated for each time $t\in\mathcal{T}$ within a specific range of active powers and power factors, similar to the procedure used in~\cite{kundu}.
A profile of available PV generation ($p_j^g(t)$) at node 6 is shown in Fig.~\ref{pg}. 



\begin{figure}
\centering
\includegraphics[width=0.9\linewidth]{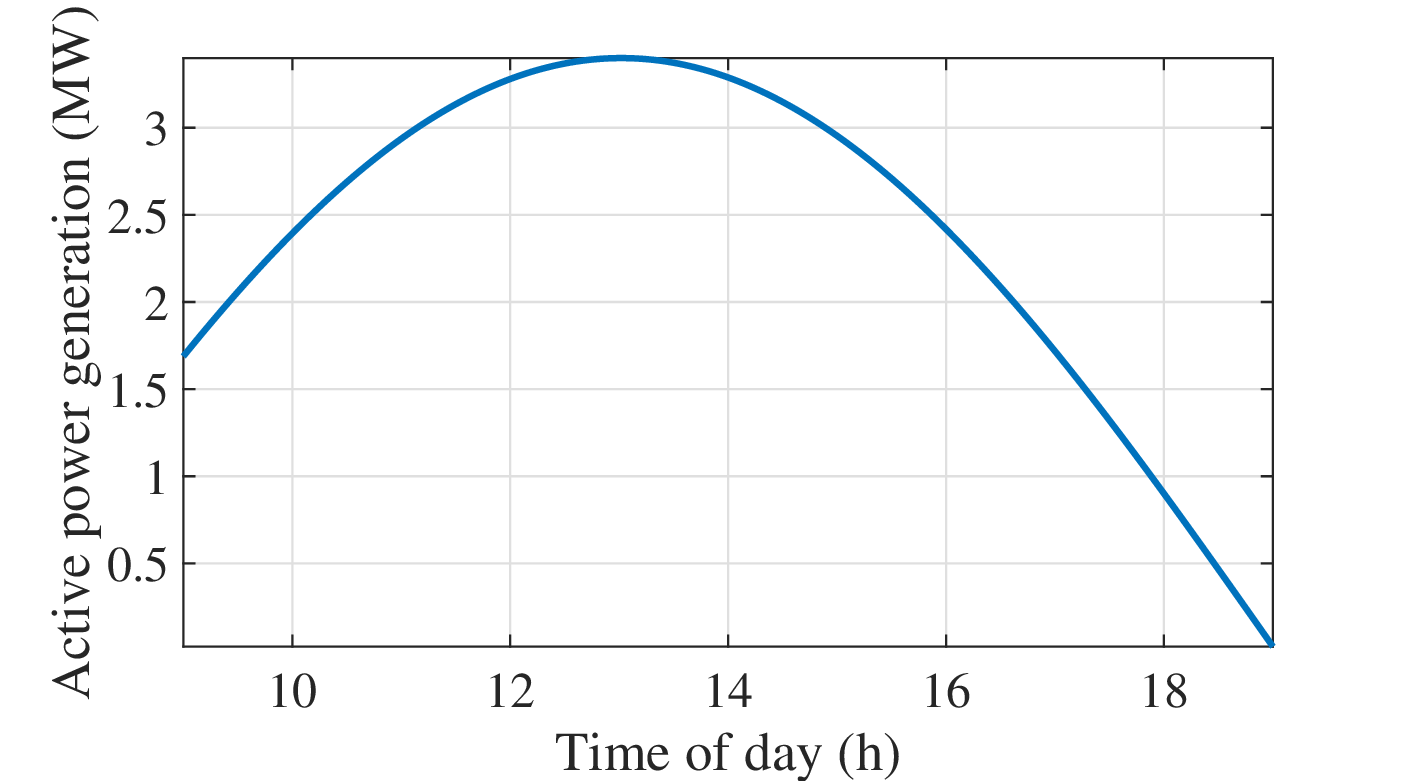}
\caption{Plot of available PV generation at node 6 during time window of interest.}
\label{pg}
\end{figure}

\begin{figure}
\centering
\includegraphics[width=0.9\linewidth]{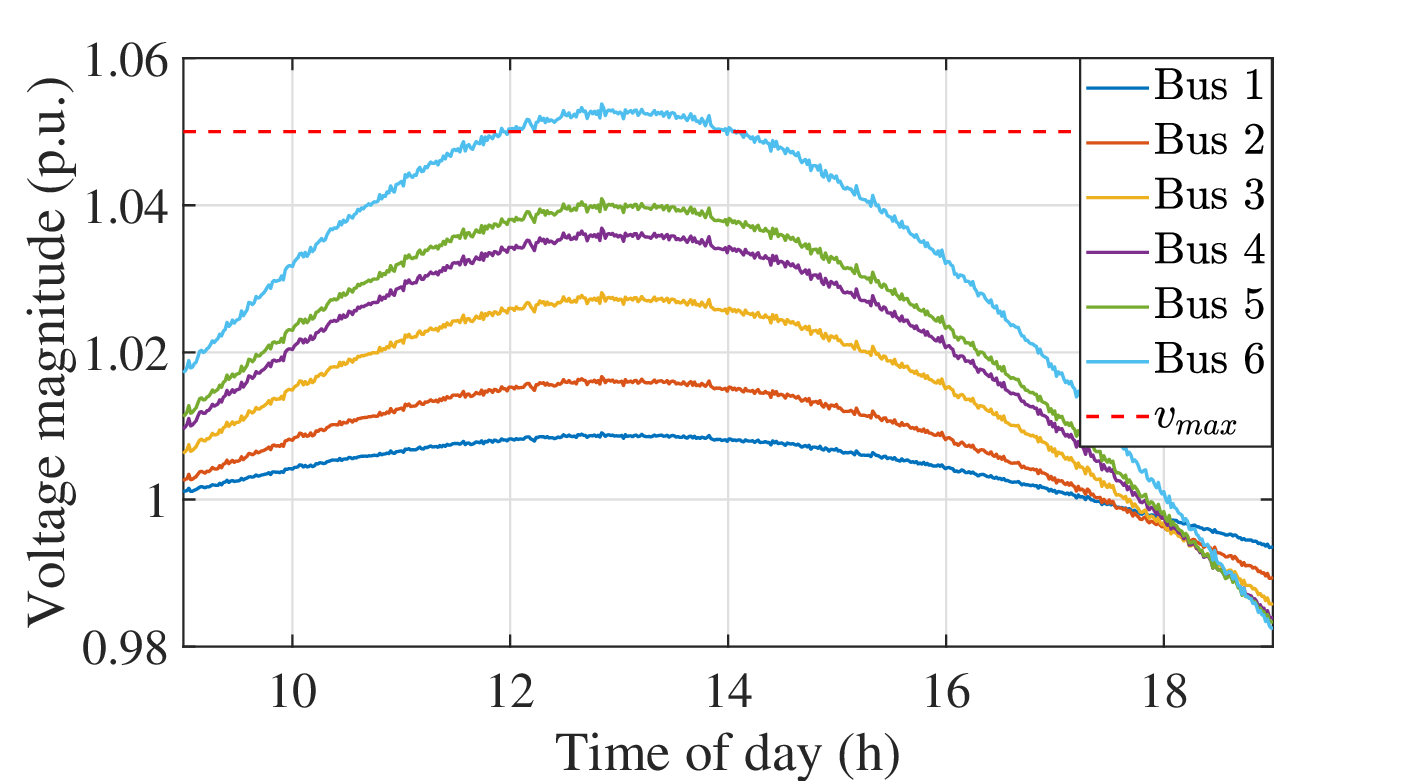}
\caption{Plot of bus voltage magnitudes in the feeder for the case when no PV curtailment is applied.}
\label{beforeVol}
\end{figure}

\begin{figure}
\centering
\includegraphics[width=0.9\linewidth]{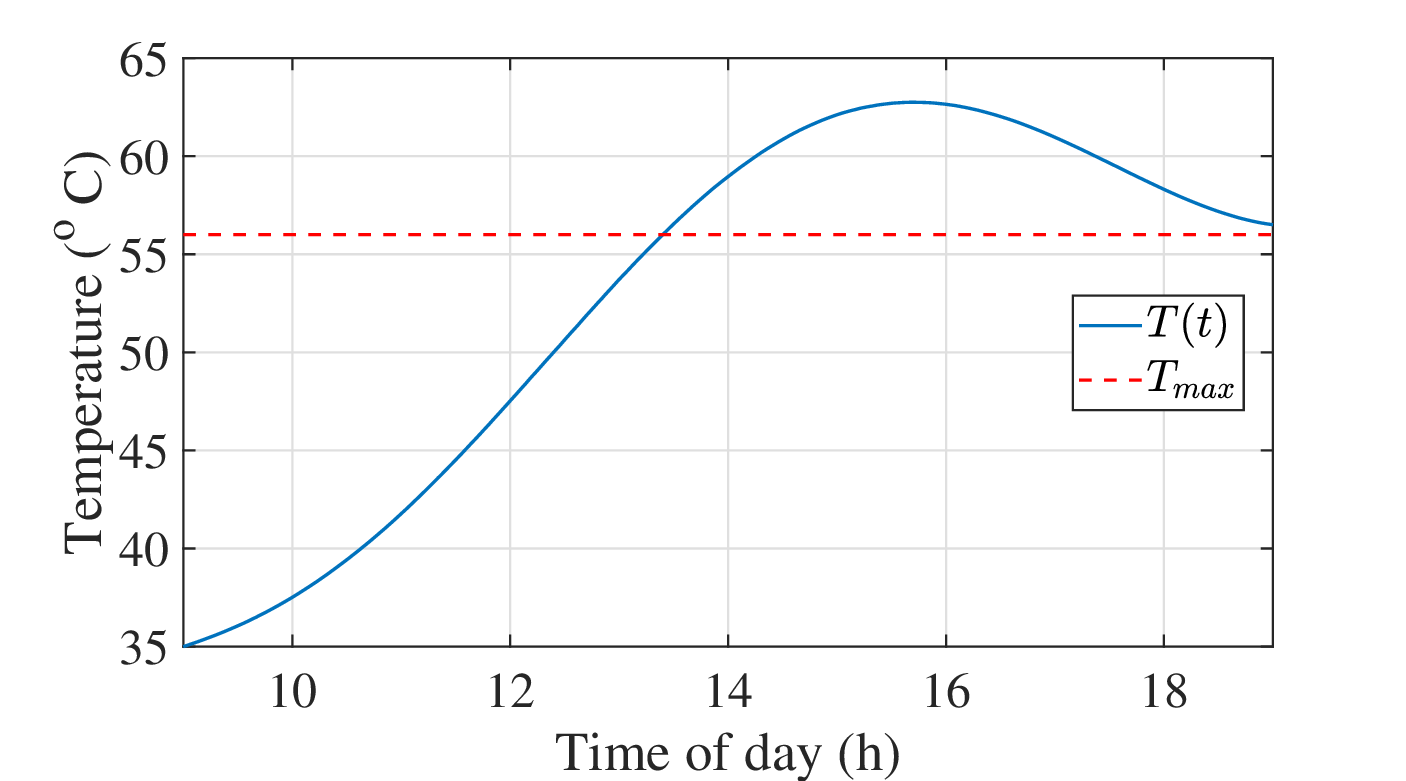}
\caption{Plot of substation transformer hot-spot temperature for the case when no PV curtailment is applied.}
\label{beforeTemp}
\end{figure}

Figures~\ref{beforeVol} and~\ref{beforeTemp} illustrate the impact of this high PV penetration when PV curtailment is not properly managed.
It can be seen from Fig.~\ref{beforeVol} that overvoltage violations occur at node 6 during peak solar output, and Fig.~\ref{beforeTemp} illustrates the rise in substation transformer temperature well above safe operating limits. 
This clearly motivates the need for control and optimization of PV curtailment.

\subsection{Impact of objective function on closed-loop performance}
\label{sec:no_curt}

Here, we numerically justify our choice of objective function in (P1).
Since our goal is to minimize PV curtailment, 
a logical first choice for the objective function would be
\begin{align}
\min_{ p_j^{cr}(h),q_j^{g}(h)} & \sum_{h=0}^{H-1} \sum_{j=1}^{N}   p_j^{cr}( h)^2 \,.  \label{ini-obj}
\end{align}
 However, using~\eqref{ini-obj} in (P1) with a short prediction horizon, $H$, can lead to poor closed-loop performance.
 This occurs because the inverter reactive power outputs, $q_j^g(h)$, become free variables (and thus depend on solver initial values) in time instances when the transformer temperature constraint is not binding during the entire prediction horizon.
 This is can be seen in Fig.~\ref{qgResult-onlyPcr} (with $H=1$~minute), where the reactive power output seems to behave erratically in the beginning of the simulation.
 When the prediction horizon is not sufficiently long enough to capture the transformer temperature dynamics accurately, these values of $q_j^g(h)$ can lead to early overloading of the transformer, as seen in Fig.~\ref{tempResult-onlyPcr}.
 This leads to the temperature constraint remaining active for a longer period of time, and requires more PV curtailment (Fig.~\ref{pcrResult-onlyPcr}).
 In order to more readily compare between cases, we introduce the total PV curtailment, $p_\text{total}^{cr}$, as a metric:
 \begin{equation}
     p_\text{total}^{cr} = \sum_{t\in\mathcal{T}}\sum_{j\in\mathcal{N}} \frac{p_j^{cr}(t)}{p_j^g(t)}\,.
 \end{equation}
 The total PV curtailment for the case when~\eqref{ini-obj} is used in (P1) is 12.4\%.

 \begin{figure}
\centering
\includegraphics[width=0.9\linewidth]{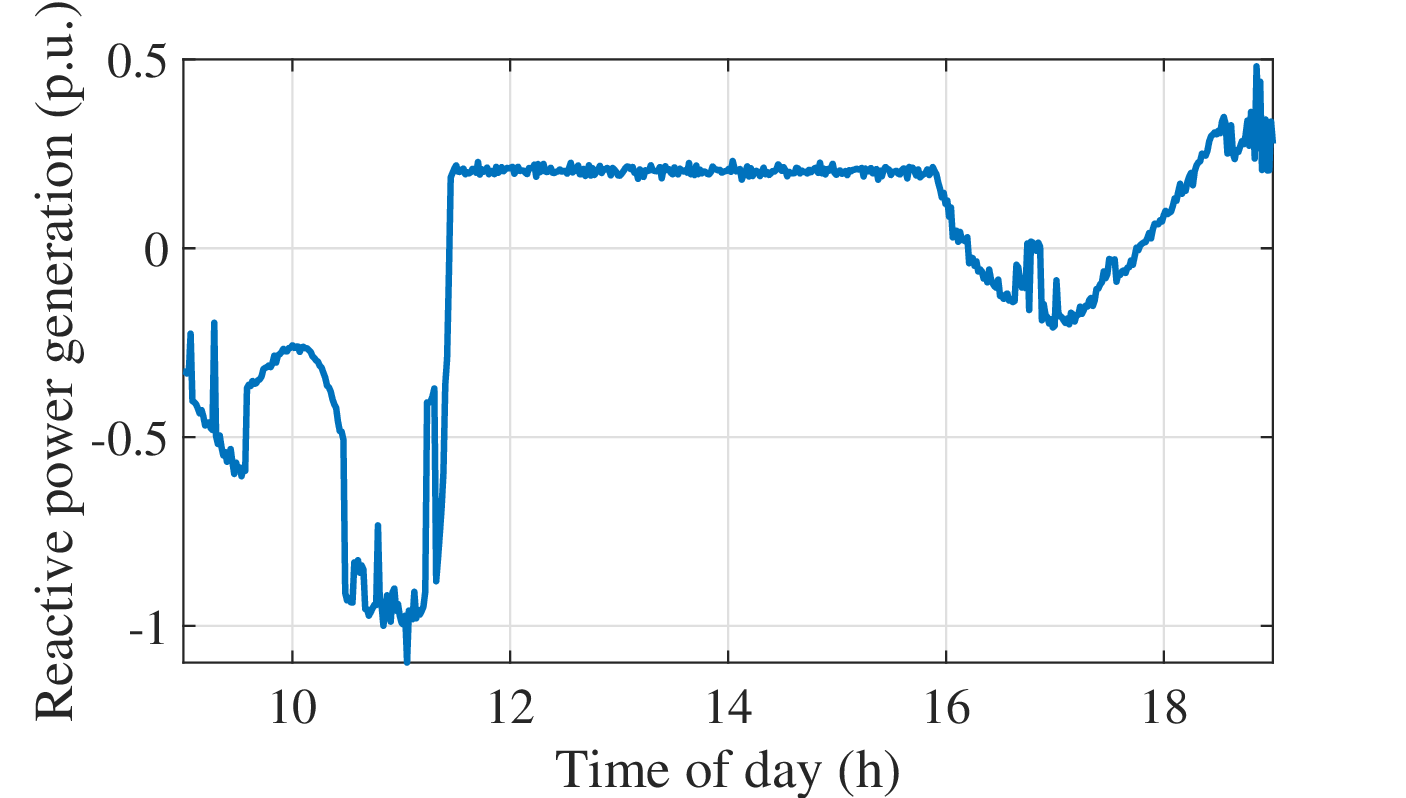}
\caption{Plot of optimal inverter reactive power output when only PV curtailment is considered in the MPC objective function.}
\label{qgResult-onlyPcr}
\end{figure}
\begin{figure}
\centering
\includegraphics[width=0.9\linewidth]{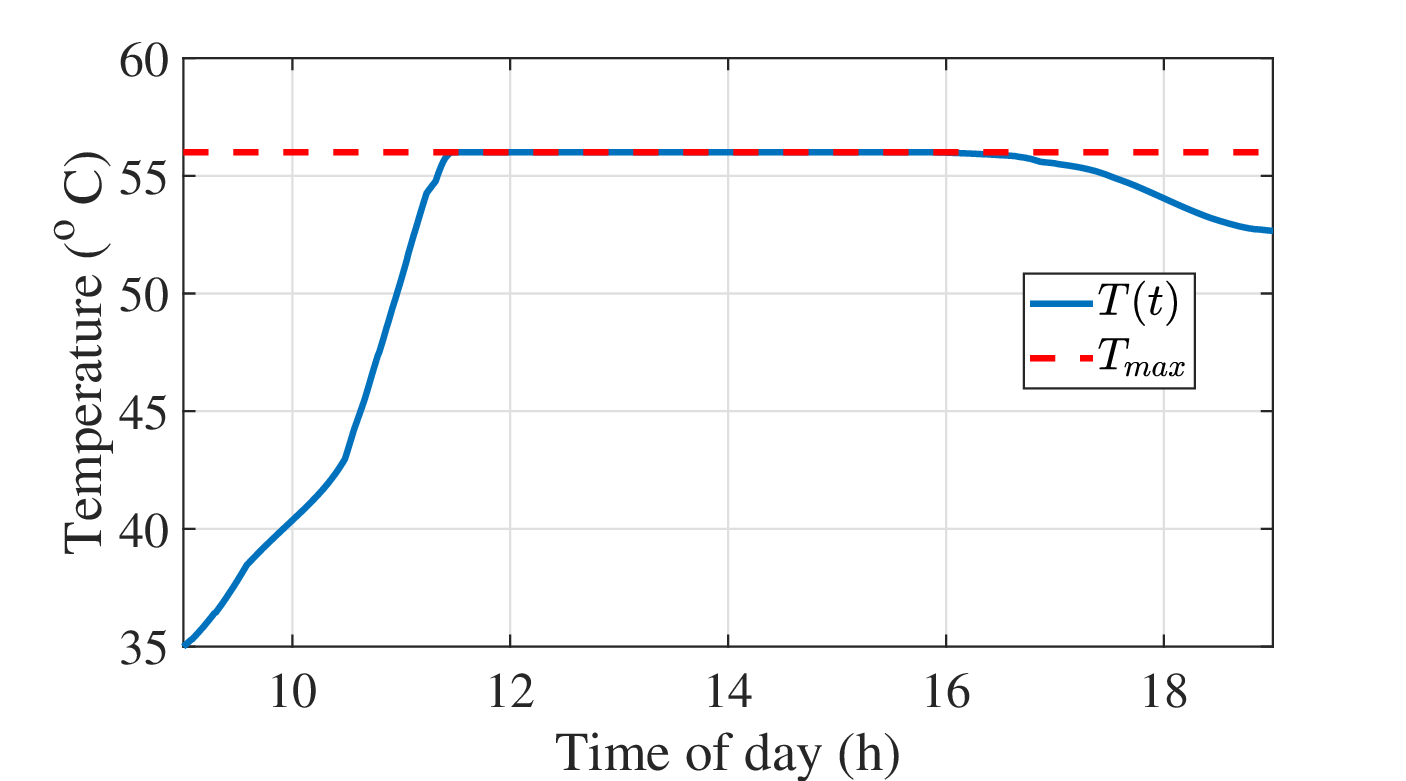}
\caption{Plot of substation transformer temperature when only PV curtailment is considered in the MPC objective function.}
\label{tempResult-onlyPcr}
\end{figure}
\begin{figure}
\centering
\includegraphics[width=0.9\linewidth]{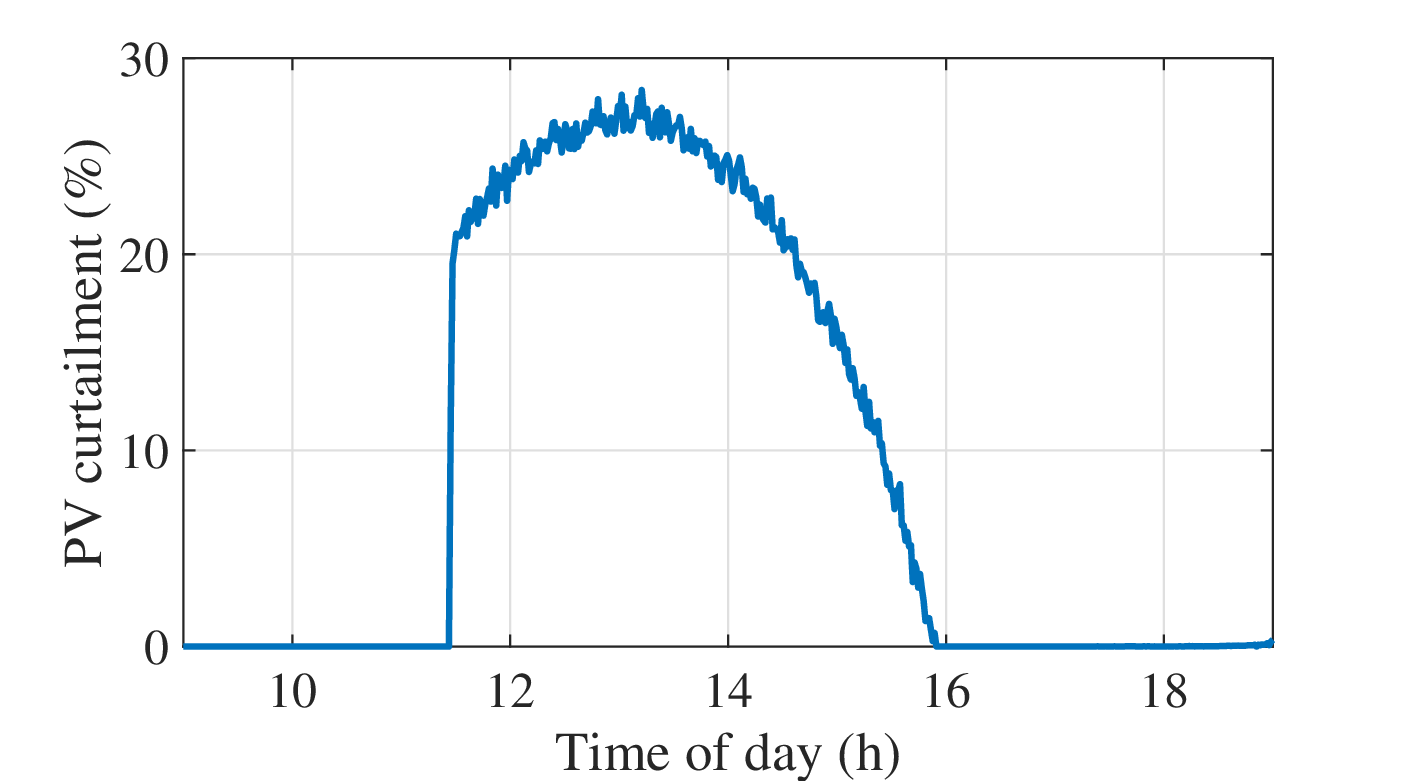}
\caption{Plot of inverter active power curtailment when only PV curtailment is considered in the MPC objective function.}
\label{pcrResult-onlyPcr}
\end{figure} 

One approach for addressing this issue could be to increase the prediction horizon such that it captures the slower temperature dynamics. 
However, increasing the prediction horizon comes at the expense of more optimization variables and a higher computational burden. 
To this end, Table~\ref{tab:comp_time} illustrates the tradeoff between total PV curtailment and average computation time required to solve (P1).\footnote{All simulations were performed on a laptop equipped with an Intel 13th Gen Core i7-1365U processor (1.80~GHz) and 32.0~GB of RAM.}

\begin{table}
\vspace{1.0em}
\centering
\caption{Impact of horizon on computation time and PV curtailment}  
\begin{tabular}{c>{\centering\arraybackslash}p{2cm}>{\centering\arraybackslash}p{2cm}>{\centering\arraybackslash}p{2cm}}
    \toprule
    $H$ & \textbf{Number of variables} & \textbf{Average comp. time} (s) & \textbf{Total curtailment} (\%) \\
    \midrule
    30  & 2402  & 2.1  & 8.0 \\
    60  & 4802  & 2.7  & 6.0 \\
    120 & 9602  & 7.3  & 4.6 \\
    \bottomrule
    \label{tab:comp_time}
\end{tabular}
\end{table}

Instead of increasing the length of the prediction horizon, we propose to 
penalize $q_j^g(h)$ in~\eqref{eq:obj_fun} to reduce reactive power setpoints that inadvertently lead to unnecessary temperature rise, as discussed previously.
In this multi-objective optimization problem, selecting appropriate values for the scaling parameter $\beta$ is crucial for ensuring the primary goal of minimizing curtailment is achieved. 
Figure~\ref{paretoANalysis} illustrates this by showing how total PV curtailment changes as $\beta$ is increased for a prediction horizon of $H=1$.
It is observed that for values of $\beta > 100$, this choice of objective function leads to performance similar to the case when $H=120$~minutes and~\eqref{ini-obj} is used.

\begin{figure}
\centering
\includegraphics[width=0.9\linewidth]{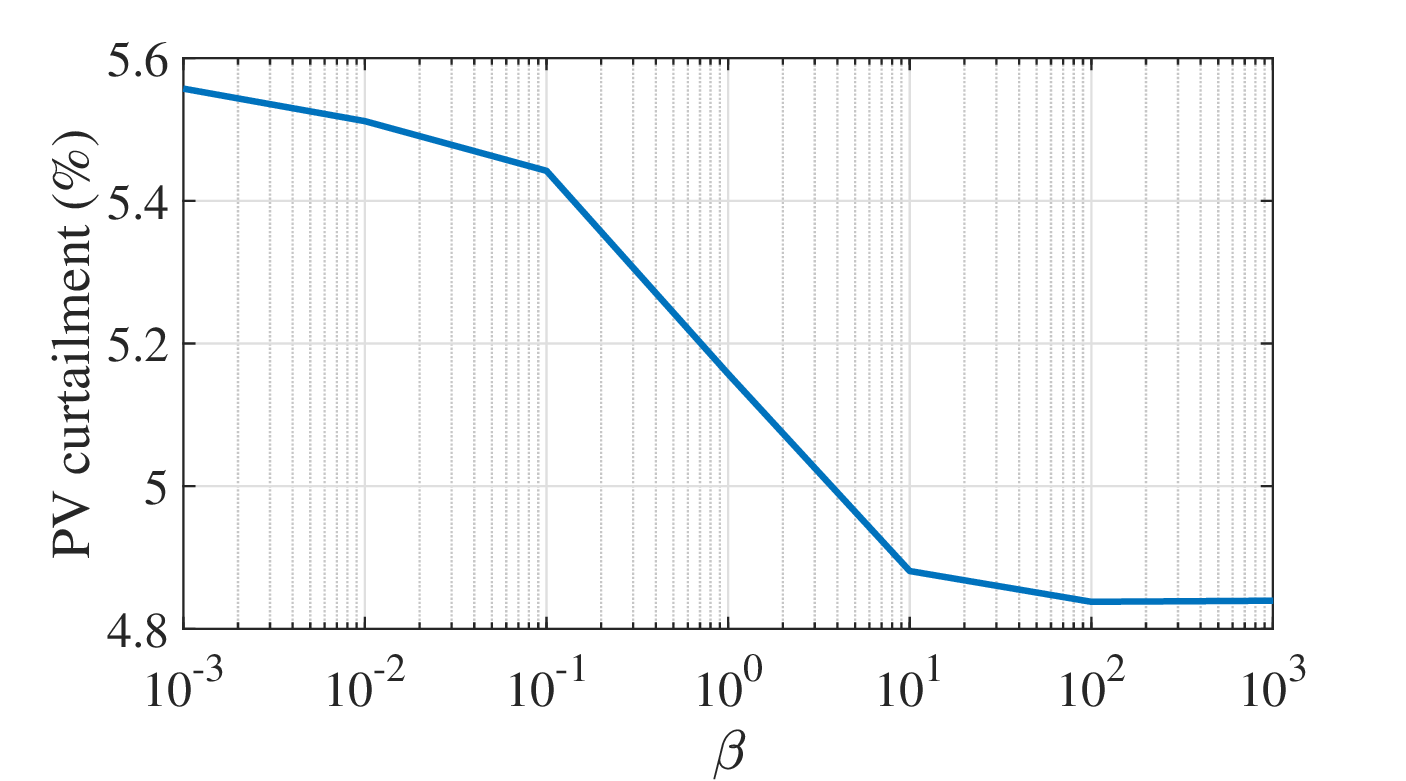}
\caption{Impact of scaling parameter, $\beta$, on total PV curtailment when $H=1$.}
\label{paretoANalysis}
\end{figure} 

\subsection{Performance of the proposed MPC approach}
Considering the proposed objective function (with $\beta=10^5$) and a prediction horizon of $H=1$, Figs.~\ref{VolResult} and~\ref{TempResult} demonstrate the successful implementation of the centralized MPC optimization strategy. 
The MPC effectively maintains bus voltage magnitudes and the substation transformer temperature within acceptable limits throughout the time window.
\begin{figure}
\centering
\includegraphics[width=0.9\linewidth]{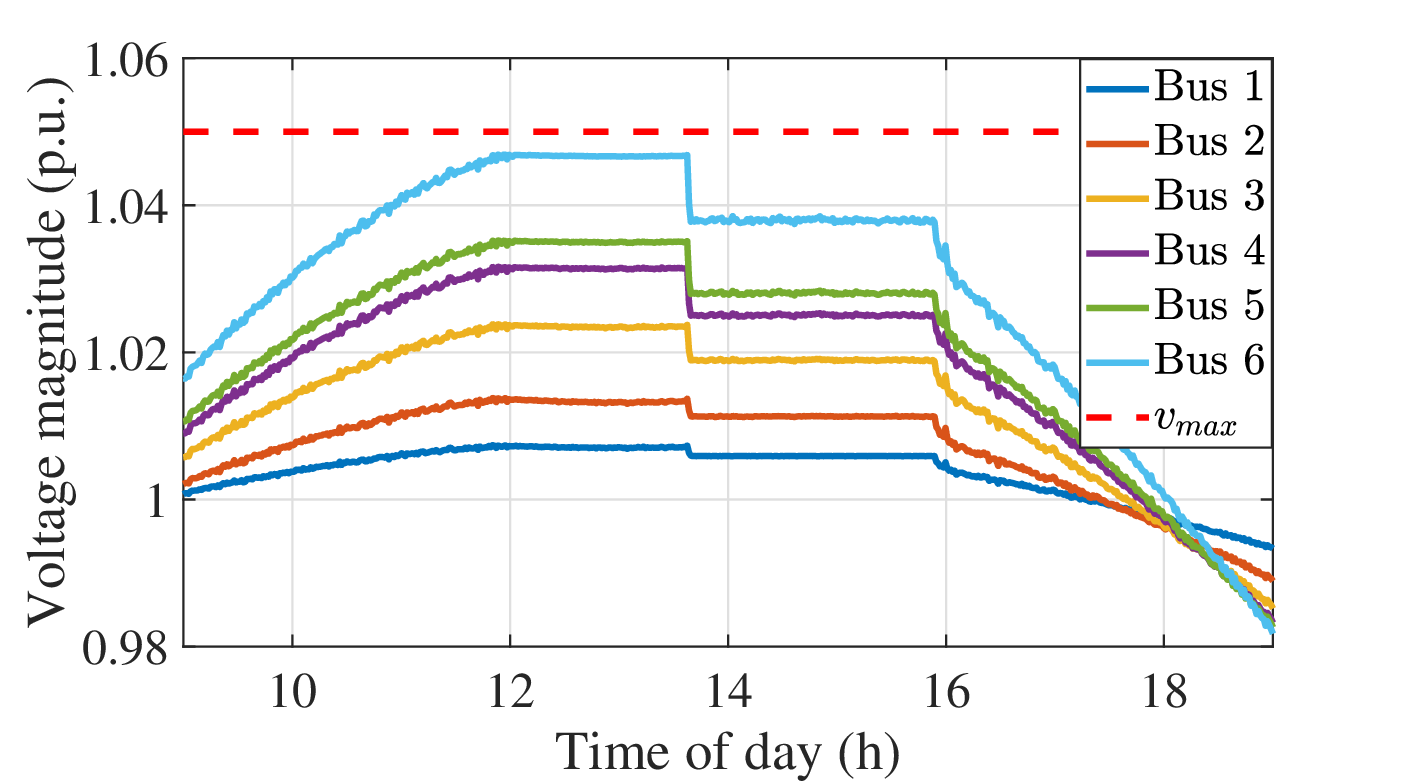}
\caption{Plot of bus voltage magnitudes in the feeder under the proposed MPC approach.}
\label{VolResult}
\end{figure}
\begin{figure}
\centering
\includegraphics[width=0.9\linewidth]{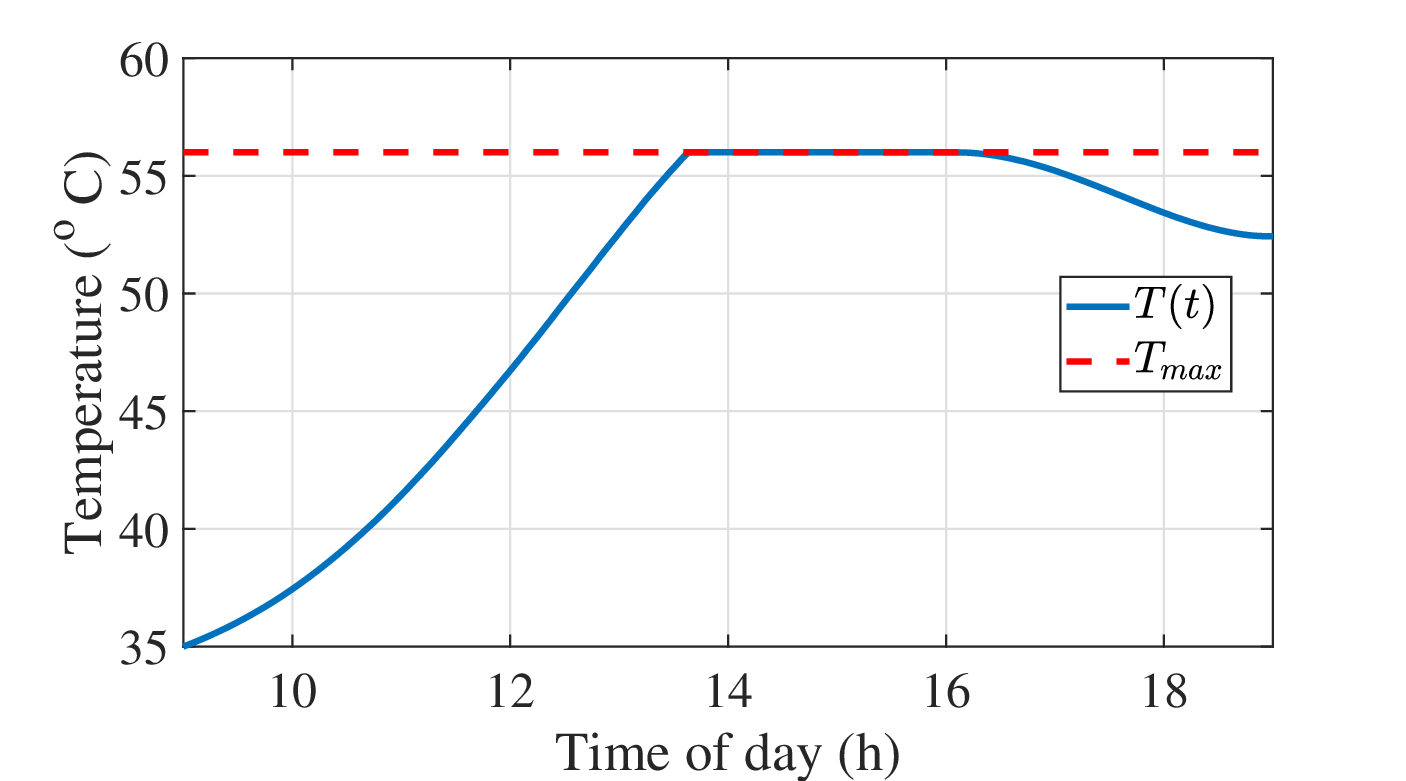}
\caption{Plot of substation transformer temperature under the proposed MPC approach.}
\label{TempResult}
\end{figure}
In particular, Fig.~\ref{TempResult} illustrates that the proposed approach prevents early overloading of transformer as the temperature constraint becomes active at $t=13$ (as compared to $t=11.5$ in the previous case). 
It is worth noting that, although it appears to be binding, the transformer temperature remains slightly below \(T_{\text{max}}\). 
This is due to model mismatch between the LinDistFlow model used in the MPC (which neglects line losses) and the AC power flow equations used in the plant model (which includes losses).
This model mismatch also explains why, in Fig.~\ref{VolResult}, the actual voltages do not hit their limits during $t=12$ to $t=13.6$ even though the predicted voltages in the MPC are binding.
This highlights the fact that the LinDistFlow equations used in the MPC always overestimates voltages and transformer temperature, which guarantees safe (yet conservative) operation despite model mismatch.

Figure~\ref{qgResult} shows the reactive power output of the inverter. 
It can be seen that the penalization of $q_j^g(h)$ leads to zero reactive power generation during the beginning of the simulation. 
As voltage constraints become binding (in the MPC prediction) around $t=12$, the MPC begins to adjust the reactive power generation of the inverter to prevent overvoltages. 
Finally, Fig.~\ref{pcrResult} depicts the curtailment of inverter active power output. 
The total PV curtailment in this case is 4.84\%, which is similar to the case when reactive power output was not penalized in the objective function and $H=120$ was used. 
This indicates that the proposed MPC is able to achieve reasonable performance with lower computational burden.
\begin{figure}
\centering
\includegraphics[width=0.9\linewidth]{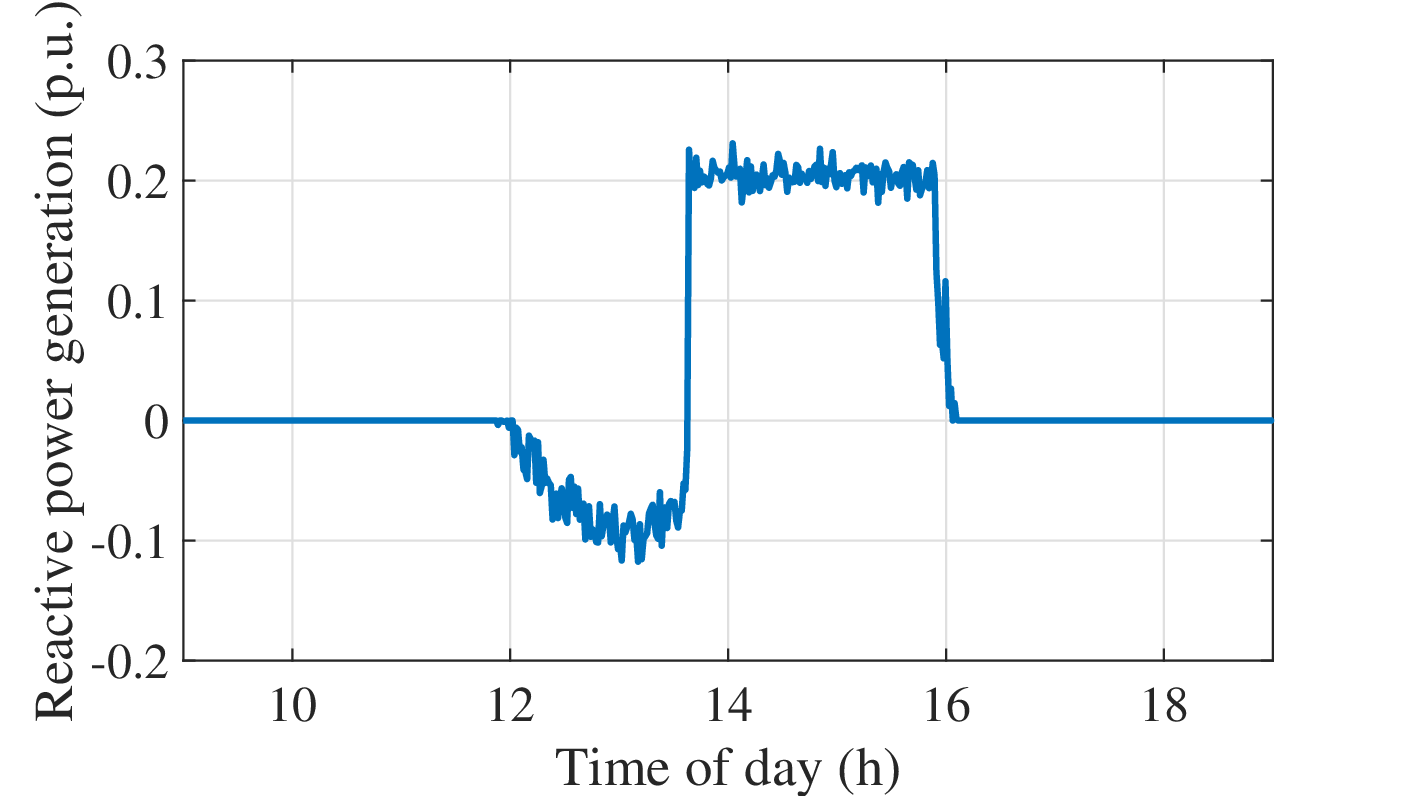}
\caption{Plot of optimal inverter reactive power output under the proposed MPC approach.}
\label{qgResult}
\end{figure}
\begin{figure}
\centering
\includegraphics[width=0.9\linewidth]{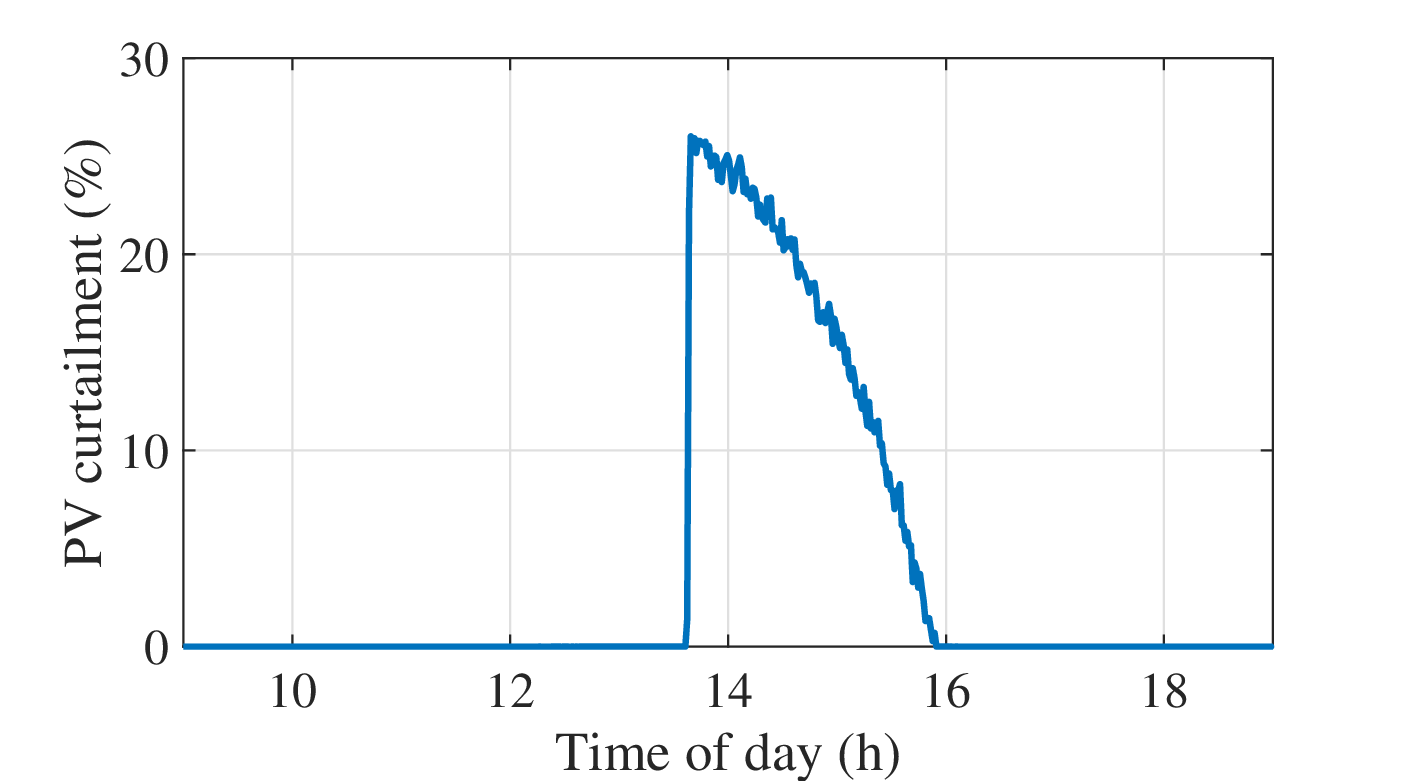}
\caption{Plot of optimal inverter active power curtailment under the proposed MPC approach.}
\label{pcrResult}
\end{figure}

\section{Conclusion and Future Work}
In this paper, a centralized MPC approach was proposed for simultaneously controlling bus voltage magnitudes and substation transformer temperatures in distribution grids with high PV penetration. 
Numerical simulations were conducted on a 6-node radial network, and the results demonstrated that both voltage and temperature could be effectively managed by the proposed controller without any constraint violations. 
Furthermore, using KKT analysis, we proved conditions under which the convex relaxation of the constraints used to model the transformer temperature dynamics is tight. 

For future work, we recognize that despite the benefits of a centralized MPC approach, real-time controller implementation can be difficult due to communication delays and computational challenges. Therefore, research into decentralized control strategies for simultaneously managing transformer temperatures and system voltages would be of value. 
However, decentralized control also comes with challenges, such as the need for effective coordination between local controllers to ensure system-wide optimality and reduce constraint violations. 
Thus, it is crucial for future work to compare the optimality of decentralized solutions with those of the centralized approach to understand the tradeoffs involved. 
Future research could also extend the KKT analysis presented herein to explore cases when voltage constraints are binding.
Finally, since solar generation and load patterns are inherently uncertain, research into the impact of imperfect forecasts on MPC performance would be of interest.
\bibliographystyle{IEEEtran}
\bibliography{Refs} 
\addtolength{\textheight}{-12cm}   




\end{document}